\documentclass[% 11pt,twocolumn
pre,
reprint,superscriptaddress
]{revtex4-1}
\usepackage{bbm}
\usepackage{xcolor}
\usepackage{enumitem}
\usepackage{amsmath}
\usepackage{amssymb}
\usepackage{amsthm}
\usepackage{graphicx}
\usepackage{hyperref}

\newtheorem{theorem}{Theorem}
\newtheorem{corollary}{Corollary}
\newtheorem{definition}{Definition}

\let\originalleft\left
  \let\originalright\right
\renewcommand{\left}{\mathopen{}\mathclose\bgroup\originalleft}
  \renewcommand{\right}{\aftergroup\egroup\originalright}

\newcommand{\second}{\mathrm{s}}                %Seconds
\newcommand{\de}{d_{\mathrm{eff}}}           %Effective Dimension
\newcommand{\id}{\mathbbm{1}}                %Identity
\newcommand{\ts}{{\otimes 2}}                %Tensor product squared
\newcommand{\expo}[1]{\operatorname{e}^{#1}} %Exponential
\newcommand{\tr}[1]{\operatorname{Tr}\left[ {#1} \right]} %Trace
\newcommand{\braket}[1]{\vphantom{\left(#1\right)^A} \left \langle #1 \right \rangle } %ket
\newcommand{\braketsmall}[1]{ \langle #1 \rangle } %ket
\newcommand{\ket}[1]{\left|#1 \right \rangle \vphantom{\left( #1 \right)^A}} %ket
\newcommand{\bra}[1]{\left\langle #1 \right | \vphantom{\left(#1\right)^A}} %bra
\newcommand{\hk}{{\mathcal{H}_K}} %\hk --> K SubSpace
\newcommand{\lorsubs}{_{L_T}}
\newcommand{\omt}{\omega\lorsubs}
\newcommand{\dif}{\mathrm{d}}
\newcommand{\teq}{T_{\mathrm{eq}}}
\newcommand{\Real}{\mathbb{R}}

\newcommand{\Hil}{\mathcal{H}}
\newcommand{\hil}{\mathcal{H}}
\newcommand{\cm}{\mathcal{M}}   %Curly M
\DeclareMathOperator\erf{erf}

\newcommand{\spn}{\operatorname{span}}
\newcommand{\rank}[1]{\relax\ifmmode\operatorname{rank}#1\else rank-$#1$\fi}
\newcommand{\Set}[2]{\left\lbrace #1 \,\middle|\, #2 \right \rbrace}
\newcommand{\set}[1]{\left\lbrace #1 \right \rbrace}

\allowdisplaybreaks[4]
\hypersetup{pdfborder={0 0 1}, bookmarks=true, linkcolor=blue}
\begin{document}

\title{Quantum Systems Equilibrate Rapidly for Most Observables}

\author{Artur S.L. Malabarba}
\affiliation{H.H. Wills Physics Laboratory, University of Bristol, Tyndall Avenue, Bristol, BS8 1TL, U.K.}
\author{Luis Pedro Garc\'{i}a-Pintos}
\affiliation{School of Mathematics, University of Bristol, University Walk, Bristol BS8 1TW, U.K.}
\author{Noah Linden}
\affiliation{School of Mathematics, University of Bristol, University Walk, Bristol BS8 1TW, U.K.}
\author{Terence C. Farrelly}
\affiliation{DAMTP, Centre for Mathematical Sciences, Wilberforce Road, Cambridge, CB3 0WA, United Kingdom}
\author{Anthony J. Short}
\affiliation{H.H. Wills Physics Laboratory, University of Bristol, Tyndall Avenue, Bristol, BS8 1TL, U.K.}
\date{\today}

% \keywords{Quantum Equilibration}
% \pacs{05.70.Ln, %Nonequilibrium and irreversible thermodynamics
% XX.XX.XX (usually up to 3 numbers) }

\begin{abstract}
    Considering any Hamiltonian, any initial state, and measurements
    with a small number of possible outcomes compared to the dimension, we show
    that most measurements are already equilibrated. To investigate
    non-trivial equilibration we therefore consider a restricted set
    of measurements. When the initial
    state is spread over many energy levels, and we consider the set of
    observables for which this state is an eigenstate, most observables are initially out of equilibrium
    yet equilibrate rapidly. Moreover, all two-outcome measurements,
    where one of the projectors is of low rank, equilibrate rapidly.
\end{abstract}
\maketitle

The topic of equilibration time scales has been of much interest
lately~\cite{ShortFarrelly11,Masanes13,Brandao12,Vinayak12,Cramer12,Kastner13,Hutter13,Goldstein13,TorSan13}.
Given that it has been shown that quantum systems equilibrate under
rather general conditions~\cite{Lin09,Reimann10,Goldstein10}, it is
important to understand the time scale for the process. However,
attempts to derive an upper bound on equilibration time have resulted
in very large time scales. Short and Farrelly~\cite{ShortFarrelly11},
for instance, obtain a very general bound, of which we give an
improved derivation in appendix~A, which
scales with the dimension of the system, typically exponentially in
the number of particles.

A quantum system is said to undergo equilibration when its quantum state spends most of its time almost indistinguishable from a
fixed (time-invariant) steady-state. This is not the same as
thermalization, in which the steady-state is a Gibbs state. Thus,
thermalization is a special case of equilibration, and understanding
equilibration times is a key step in understanding
thermalization times.

When one discusses quantum equilibration, it is common to refer to
either subsystem equilibration~\cite{Lin09,Lin10}, in which a small
system equilibrates due to contact with a bath, or observable
equilibration, in which a fully closed system appears to equilibrate
due to the limited information offered by outcomes of a particular set
of observables. The latter was initially shown by
Reimann~\cite{Reimann08,Reimann10}, as a statement that the
expectation values of quantum observables stay predominantly close to
a static value, and was later built-on by Short~\cite{Short11}, who
showed that these results apply even if one considers all the
information that can be gathered from the observable, instead of just
the expectation value.

In this paper we consider any finite-dimensional system and any
Hamiltonian, and show that most $N$-outcome observables are initially
in equilibrium (for $N$ small compared to the dimension). To
investigate timescales we therefore turn to a natural class of
observables which are initially typically out of equilibrium -- those
with a definite initial value (i.e. observables for which the initial
state is an eigenstate). We show that, for pure initial states spread
over many energy levels, most of these observables equilibrate in very
short times -- in fact, most equilibrate essentially as fast as
possible. Moreover, in the case of two-outcome observables where one
of the projectors is of low rank we show that all observables
equilibrate fast (for any initial state spread over many energy
levels).

As will be clear in theorems~\ref{thr:3} and~\ref{thr:5}, when
referring to ``typical'' or ``most'' observables we mean that in the
context of the Haar measure. While this does include all observables
of physical significance, they constitute a small fraction of all
possible observables. Still, these results give us new insight into
the problem of equilibration time scales -- which was not available
through the use of strict upper bounds. This also raises the question
of what is special about physical measurements that makes them much
slower than most measurements.

To obtain these results, we address the issue of equilibration time
scales with respect to measurements composed of $N$ outcomes. As a
figure of merit for equilibration we will use the
\emph{distinguishability} $D_\mathcal{M}(\sigma,\rho)$ between two
states $\sigma$ and $\rho$ according to an observable $\mathcal{M} =
\{P_1, \ldots, P_N \}$, where the projectors $P_j$ represent the
different outcomes of the measurement. We define it so that after
performing the measurement, given full information about the two
states being compared, the distinguishability quantifies the
probability of successfully ``guessing'' which state the system was
in~\cite{Short11}, according to
\begin{equation}
    \label{eq:11}
    p_{\text{success}} = \frac{1}{2} + \frac{1}{2}D_\mathcal{M}(\sigma,\rho),
\end{equation}
where
\begin{equation}
    \label{eq:8}
    D_\mathcal{M}(\sigma,\rho) = \frac{1}{2} \sum_j |\! \tr{\sigma P_j} - \tr{\rho P_j} |.
\end{equation}
When $D_\mathcal{M}(\sigma,\rho) = 0$ the measurement does not provide
information that helps to distinguish $\sigma$ from $\rho$ (one might
as well toss an unbiased coin to decide). On the other hand, when
$D_\mathcal{M}(\sigma,\rho) = 1$ the states are perfectly
discriminated by this measurement.

In the special case in which the measurement has two outcomes (i.e.
$\mathcal{M} =\{P, \id-P\}$), the distinguishability is given by
$D_{\mathcal{M}}(\sigma, \rho) = \left|\tr{\sigma P} - \tr{\rho
    P}\right|$, and we will denote it by $D_{P}(\sigma, \rho)$.

Given an initial state $\rho$ evolving under a Hamiltonian $H$, we say
equilibration has taken place at time $\teq$ when, for some small
constant $\varepsilon>0$,
\begin{equation}
    \label{eq:28}
    \braket{D_\mathcal{M}(\rho_t,\omega) }_T \leq \varepsilon,\quad \forall T> \teq,
\end{equation}
where $\braket{f(t)}_T = \frac{1}{T} \int_{0}^{T} f(t) \dif t$,
$\omega = \lim_{T \rightarrow \infty} \braket{\rho_t}_T$ is the
time-averaged state and $\rho_t = \expo{-iHt}\rho\expo{iHt}$ is the
evolved state. When equilibration does take place, $\omega$ is also
called the equilibrium state~\footnote{Note that this limit is
  well-defined, and is obtained by decohering $\rho$ in the energy
  eigenbasis. i.e. $\omega = \sum_n \mathcal{Q}_n\rho \mathcal{Q}_n$
  where $\mathcal{Q}_n$ is the projector onto the $n$-th energy
  eigenspace.}.

In words, equilibration with respect to a measurement has taken place
when the time-averaged distinguishability for this measurement falls
below a small value. Since the distinguishability is a positive
quantity this means that, for any $T>\teq$, the instantaneous state $\rho_t$ is essentially
indistinguishable from $\omega$ for almost all times in the interval $[0,T]$.

It is worth stating that we are not claiming thermalization, in which
the equilibrium state is a thermal state. This definition of
equilibration, while necessary for thermalization~\cite{Lin09}, is a
separate issue and guarantees only the approach to a steady state.

Two parameters of the initial state that will be of importance to us
are the effective dimension $\de$ and the energy standard deviation
$\sigma_E$, defined by
\begin{equation}
    \label{eq:26}
    \de = \left[ \sum_{n=1}^{\tilde{d}} p_n^2 \right]^{-1} \quad\quad\quad \sigma_E^2 = \sum_{n=1}^{\tilde{d}} p_n \left( E_n - \bar{E} \right)^2.
\end{equation}
Here $\tilde{d}$ is the number of distinct energy levels $E_n$ with
probabilities $p_n = \tr{\mathcal{Q}_n \rho}$, where $\mathcal{Q}_n$
is the projector onto the eigenspace with energy $E_n$, and $\bar{E} =
\sum_{n=1}^{\tilde{d}} p_n E_n$. The effective dimension gives an
estimate of how many energy levels $\rho$ occupies with significant
probability \footnote{In particular, if the state is spread evenly
  over $N$ energies then $\de=N$}, and $\de \gg 1$ is the key
requirement for equilibration to take place, as we will see. The
energy standard deviation, in turn, will take a primary role in the
time scale of equilibration. Note that we use $\hbar=1$ throughout.

Given any $d$-dimensional Hamiltonian and any initial state with
support on many energy levels (high effective dimension), we present
the following results regarding the time scale of equilibration:
\begin{enumerate}
   \item%
    we show that \emph{all} two-outcome measurements, for which one of
    the projectors $P$ has small rank, have very short time scales
    ($\teq \sim \rank{P}/\sigma_E$);
   \item%
    for $N$-outcome measurements, we prove that \emph{most}
    measurements are already equilibrated, and that \emph{most}
    measurements with a definite initial value (i.e. for which the
    initial state is a pure eigenstate) equilibrate with extremely
    short time scales ($\teq \sim 1/ \sigma_E$). These results hold as
    long as $N \ll d$ but regardless of the rank of the projectors.
\end{enumerate}

% We will prove these statements, in this order, in the following two
% sections.
The first statement shows how restricting the measurements can lead to
a whole family of observables that equilibrate fast, while the second
one refers to the time scales of typical measurements (under the Haar
measure), which are indeed also fast.

To understand why $1/{\sigma_E}$ is a fast time scale, one can
consider the uncertainty relation between $H$ and an observable $O$,
which states~\cite{merzbacher1998quantum} $2\sigma_E \sigma_O \geq
|\braket{[H,O]}| = |{\braketsmall{ \dot O}}|$. Thus, the minimum
necessary time for the expectation value of \emph{any} observable to
vary significantly is ${\sigma_O}/{|{\braketsmall{ \dot O}}|} \geq 1/
(2\sigma_E)$.

Previous related results on equilibration time scales have been
obtained by Goldstein, Hara and Tasaki~\cite{Goldstein13}. In that
paper the authors considered particular constructions of two-outcome
measurements that take a very long or very short time to equilibrate.
Interestingly, they find an example of a projector with high rank
which equilibrates fast, albeit with very particular properties. Our
first result proves fast equilibration of \emph{all} small rank
projectors, for \emph{any} system with high $\de$. Moreover, we find
that most observables also equilibrate fast (when $N\ll d$ and initial
state is pure with high $\de$). On the issue of slow time scales,
we give an alternative example to that in~\cite{Goldstein13} of a
measurement which equilibrates slowly ($\teq~\gtrsim~\de/\sigma_E$),
given a pure initial state with high $\de$.
% Both examples involve projectors of similar rank and time scales,
% yet it is not clear to us whether the two constructions are at all
% similar, which might be interesting to study in future works.

We end with a short discussion. Our results can also be stated almost
identically in terms of expectation value of observables
(distinguishability is used here as it provides a stronger statement
of equilibration).

\section{Fast Equilibration}
\label{sec:fast-equil-observ}
We now show that all systems with high effective dimension equilibrate
fast with respect to the two-outcome measurement $\mathcal{M} = \{P,
\id - P \}$, whenever either of the projectors is of sufficiently low rank.
Given  $K=\min \left\{ \rank{P},\rank{(\id-P)} \right\}$, we will show
 that the average distinguishability $\braket{D_P(\rho_t,\omega)}_T$ becomes small in a time of  the order
of $K/\sigma_E$. We start by defining the function $\eta_\epsilon$:
\begin{definition}
    \label{def:1}
    Given any Hamiltonian with spectrum $\{E_j \,|\,
    j=1,\ldots,\tilde{d}\}$, a $d$-dimensional Hilbert space $\hil$
    and any state $\rho: \Hil \rightarrow \hil$ with probabilities
    $p_j$ associated to each energy level. For any $\epsilon > 0$
    \begin{equation}
        \label{eq:12}
        \eta_\epsilon = \max_{E \in \Real} \!\!\!\! \sum_{\substack{j:\\ E_j \in [E,E+\epsilon]}}\!\!\!\! p_j
    \end{equation}
    is the maximum probability that can be found inside any energy
    interval of size $\epsilon$.
\end{definition}
This function is useful because it captures both the state's energy
distribution and the Hamiltonian's energy spectrum.

The theorem is then simply an upper bound on the finite-time average
of the distinguishability $D_P(\rho_t,\omega)$ for any projector $P$.
\begin{theorem}[Fast equilibration]
    \label{thr:4}
    For any initial state $\rho: \Hil \rightarrow \Hil$, any Hamiltonian, and any
    projector $P$ where $K=\min \left\{ \rank{P},\rank{(\id-P)} \right\}$,
    \begin{equation}
        \label{eq:13}
        \braket{D_P(\rho_t,\omega)}_T \leq c\sqrt{\eta_{\frac{1}{T} } K },
    \end{equation}
    where $c = \frac{5\pi}{4}\sqrt{\frac{2}{ {1- \expo{-2}}} } +1 \approx 6.97$
\end{theorem}
\begin{proof}
    Defining the Lorentzian average $\langle f(t) \rangle_{L_T} =
    \int_{-\infty}^{\infty} \frac{f(t)T}{T^2 +
      (t-T/2)^2}\frac{dt}{\pi}$, any positive function $f$ satisfies
    $\langle f \rangle_T \le \frac{5 \pi}{4} \langle f \rangle_{L_T}
    $~\footnote{Define \unexpanded{$\Theta_T(t) = \frac{1}{T}$} for
      \unexpanded{$t \in [0,T]$}, and \unexpanded{$0$} otherwise. Then
      \unexpanded{$\braket{f}_T = \int_{- \infty}^{\infty}
        \Theta_T(t)f(t) \dif t$}, and \unexpanded{$\Theta_T(t) \leq
        \frac{5}{4} \frac{T}{T^2 + (t-T/2)^2}$}. }.
    Then by use of the Cauchy-Schwarz inequality and the fact that
    $\tr{\omega^2} \le 1/\de$~\footnote{The equality is easy to check
      for non-degenerate Hamiltonians or pure initial states, while
      the inequality is necessary for degenerate Hamiltonians with a
      mixed initial state.}
    \begin{align}
        \label{eq:15}
        \braket{ D_P(\rho_t,\omega)}_T&= \braket{\left| \tr{P(\rho_t-\omega)} \right|}_T \notag\\
        &\leq \braket{ \tr{P\rho_t}}_T + \tr{P\omega} \notag\\
        &\leq \frac{5\pi}{4} \tr{P\omt} + \sqrt{\tr{\omega^2} \tr{P^2}} \notag\\
        &\leq \frac{5\pi}{4} \sqrt{\tr{\omt^2} \tr{P^2}} + \sqrt{\tr{\omega^2} \tr{P^2}} \notag\\
        &\leq \frac{5\pi}{4} \sqrt{K \tr{\omt^2}} + \sqrt{\frac{K}{d_{\text{eff}}}},
    \end{align}
    where $\omt = \langle \rho_t \rangle_{L_T}$ and $K=\rank{P}$.
    Appendix~B shows that
    \begin{equation}
        \label{eq:17}
        \tr{\omt^2} \leq \frac{2\eta_{\frac{1}{T} }}{1-\expo{-2}}.
    \end{equation}
    Using the fact that $\de^{-1} \leq p_{\max{}} \leq \eta_\epsilon,
    \forall \epsilon>0$, where $p_{\max{}}$ is the maximum occupation
    probability of any energy level, results in eq.~(\ref{eq:13}). The
    reason we may take $K = \min \left\{ \rank{P},\rank{(\id-P)}
    \right\}$ is that $D_P(\rho_t,\omega) = D_{\id - P}(\rho_t,\omega)
    \,\,\forall t \in \Real$.
\end{proof}
The requirement of ``large $\de$'' mentioned in the introduction is a
consequence of $\de^{-1} \leq \eta_\epsilon$, since $\eta_\frac{1}{T}$
cannot converge to a small value if $\de$ is small.

Note that the quantity $\tr{\omt^2}$, corresponding to the purity of
the time-averaged state, is at the core of the equilibration process,
dictating, for any given system, an upper bound on the timescale of
equilibration (the reciprocal of this quantity acts like a
time-dependent effective dimension, growing from $1$ to $\de$ as $T$
increases). The right hand side of eq.~(\ref{eq:17}) displays a bound
on the purity which is easier to calculate than the purity itself (in
fact, it is trivial if one knows the spectrum and the state) and whose
tightness is discussed below.

This theorem proves that the measurement of a \rank{1} projector
equilibrates as soon as the energy interval $1/T$ is too small to
contain a significant portion of the probabilities, which happens
roughly when it is small compared to $\sigma_E$ (which is a very short
time scale). Conversely, a \rank{K} projector requires that the
probabilities be $K$ times smaller. For instance, if $\eta_\frac{1}{T}
\sim \frac{1}{\sigma_E T}$, as we argue below, the time scale of
equilibration is at most $\sim \frac{K}{\sigma_E} $.

\subsection{Estimating $\eta$}
\label{sec:esimating-eta-2}

We focus now our attention on equation~(\ref{eq:17}), in order to
compare how well $\eta_{\frac{1}{T} }$ upper-bounds the purity, and to
illustrate how easy it is to estimate $\eta$.

% We argue that $\eta_{\frac{1}{T} } \sim \frac{1}{\sigma_E T} $ and
% show that it holds for an example.

Given a dense enough energy spectrum, we can approximate the
probability distribution of the initial state by a continuous function
$p(E)$ for which the maximum value is roughly
\begin{equation}
    \label{eq:18}
    \max_E{p(E)} \sim \frac{a}{\sigma_E} ,
\end{equation}
where $a$ is some constant which depends on the shape of the
distribution. Since $\eta_\epsilon$ can always be upper bounded by
$\epsilon \max_E{p(E)}$, we have
\begin{equation}
    \label{eq:19}
    \eta_{\frac{1}{T} } \le \frac{a}{\sigma_E T} ,
\end{equation}
as long as $T$ is not large enough that the $\frac{1}{T} $ window only
contains a few energy levels.
% (a trivial restriction, given that equilibration will happen while
% $T$ is still small).
In appendix~B1 we show that the above
estimation is correct for the case of a Gaussian distribution for the
energy probabilities, with $a \approx 0.40$ in this case.

\section{Typical Measurements}
\label{sec:average-projectors}
Here we prove two statements regarding typical two-outcome
measurements composed of a projector of \emph{any} rank, applied to
any initial state and any Hamiltonian.
\begin{theorem}[Typical two-outcome measurements are already
    equilibrated for any initial state]
    \label{thr:3}
    Take the \rank{K} projector $P_U$ defined as the unitary
    transformation from an energy basis projector
    \begin{equation}
        \label{eq:236008}
        P_U = UPU^\dagger = \sum_{n = 1}^{K} U \ket{n}\bra{n} U^\dagger,
    \end{equation}
    with $U:\Hil \rightarrow \hil$ unitary and $\ket{n}$ being energy
    eigenstates.

    The distinguishability between $\rho_t$ and $\omega$ according to
    $P_U$ (and its complement) averaged over all unitaries is
    \begin{equation}
        \label{eq:33}
        \braket{D_{P_U}(\rho_t,\omega)}_U
        \leq \sqrt{ \frac{K}{d^2} \frac{d-K}{d+1}}
        \leq \frac{1}{2\sqrt{d+1}}.
    \end{equation}
\end{theorem}
\begin{proof}
    The only necessary inequality is the first step, Jensen's
    inequality\cite{rudin1987real},
    \begin{equation}
        \label{eq:139501}
        \braket{D_{P_U}(\rho_t,\omega)}_U \leq \sqrt{\braket{D_{P_U}(\rho_t,\omega)^2}_U}.
    \end{equation}
    In appendix~C1 we show that the average
    of the squared distinguishability can be exactly calculated to be
    \begin{equation}
        \label{eq:30}
        \braket{D_{P_U}(\rho_t,\omega)^2}_U
        = \frac{K}{d} \frac{d-K}{d^2-1}\tr{\rho_t^2-\omega^2}.
    \end{equation}
    Then, the fact that $\tr{\omega^2} \geq 1/d$~\footnote{It is easy to
      see the trace is minimized when $\omega = \id/d$.} implies that
    $\tr{\rho_t^2-\omega^2}\leq(d-1)/d$ and leads to the first inequality
    of eq.~(\ref{eq:33}). The second inequality is obtained by setting $K
    = d/2$, which maximizes the expression.
\end{proof}
This average result is relevant because $D_{P_U}(\rho_t,\omega)$ is a
positive definite quantity. Thus, stating that its average is small
necessarily implies that $D_{P_U}(\rho_t,\omega)$ is small for most
$P_U$ (in other words, it is strongly concentrated close to zero).

This result, however, does not make any statements about time scales,
or the dynamics of equilibration. It is more relevant to study
measurements which start out of equilibrium, and ask how fast they
approach it. For this reason, Theorem~\ref{thr:5} visits again the
average distinguishability, but constrains the projector to contain
$\rho(0) = \rho_0 = \ket{\Psi}\bra{\Psi}$ as one of its terms (note
that for this theorem, we restrict the initial state to be pure). For
that, we divide the Hilbert space between the span of the initial
state and everything else $\Hil = \Hil'\oplus \rho_0$, where
$\dim{\Hil}=d$ and $\dim{\Hil'}=d-1$.

\begin{theorem}[Typical two-outcome measurements with a definite
    initial value equilibrate fast for any pure initial state with
    high $\de$]
    \label{thr:5}
    Consider the projector given by
    \begin{align}
        \label{eq:287712}
        \Pi_U = \rho_0 + P_U, \qquad \quad P_U = U PU^\dagger,
    \end{align}
    where $\rho_0$ is the initial (pure) state, $U$ is a partial unitary
    with $UU^\dagger = U^\dagger U = \id_{\Hil'}$, $P$ is any
    rank-$(K-1)$ projector with support on $\hil'$.

    The distinguishability between $\rho_t$ and $\omega$ according to
    $\Pi_U$ (and its complement) averaged over all unitaries on $\hil'$
    is
    \begin{equation}
        \label{eq:34}
        \braket{D_{\Pi_U}(\rho_t,\omega)}_U
        \leq D_{\rho_0}(\rho_t,\omega) + \frac{1}{2 \sqrt{d-1}},
    \end{equation}
    where we have (from Theorem~\ref{thr:4}) that
    \begin{equation}
        \label{eq:38}
        \braket{\braket{D_{\Pi_U}(\rho_t,\omega)}_U}_T
        \leq c \sqrt{\eta_{\frac{1}{T} }}
        + \frac{1}{2 \sqrt{d-1}}
    \end{equation}
    decays very fast, with $c \approx 6$.
\end{theorem}
% \begin{proof}
%     See appendix~\ref{sec:average-proj-cont-rho}.
% \end{proof}
The details of the proof can be found in appendix~C2. This result is enough to state that most
two-outcome measurements (of any rank) containing the initial state
equilibrate essentially as fast as the measurement of the rank-$1$
projector consisting of \emph{only} the initial state.
%%%%%%%%%%%%%%%%%%%%%%%%%%%
% If we focus on the details of the calculations we can also draw
% other conclusions. For instance, the penultimate line of
% eq.~\eqref{eq:471969} shows that increasing the rank of $P$ actually
% makes $\Pi_U$ worse at identifying $\rho_0$.

To show that this class of observables is typically out of equilibrium
initially (and thus equilibrates in a non-trivial way), we show in
appendix~C3 that the average initial distinguishability is given by
\begin{equation}
    \braket{D_{\Pi_U}(\rho_0,\omega)}_U \geq \left(1- \frac{K-1}{d-1} \right) \left(1- \frac{1}{\de} \right),
\end{equation}
and is therefore significantly above zero so long as the projector
does not cover almost the entire space.

We now extend these results to multi-outcome measurements.

\begin{corollary}[$N$-outcome generalization of Theorem~\ref{thr:3}]
    \label{co:1}
    Given $N \ll d$, the typical $N$-outcome measurement is already
    equilibrated. Describing the measurement by the POVM
    $\mathcal{M}_U = \set{U^\dagger P_i U}_{i=1,N}$, and using the
    result from Theorem~\ref{thr:3}, it is easy to see that
    \begin{align}
        \label{eq:44}
        \braket{D_{\mathcal{M}_U}(\rho_t,\omega)}_U
        &= \frac{1}{2} \sum_{j=1}^{N} \braket{D_{U^\dagger P_j U}(\rho_t,\omega)}_U \notag\\
        &\leq \frac{1}{2} \sum_{j=1}^{N} \sqrt{ \frac{K_j}{d^2} \frac{d-K_j}{d+1}} \notag\\
        &\leq \frac{1}{2} \sqrt{\frac{N}{d+1} }
    \end{align}
    where $K_j = \rank{P_j}$, and the second line is maximal for $K_j =
    d/N$.
\end{corollary}
\begin{corollary}[$N$-outcome generalization of Theorem~\ref{thr:5}]
    \label{co:2}
    Given $N \ll d$, typical out-of-equilibrium $N$-outcome
    measurements equilibrate fast for any pure initial state with high
    $\de$. Define the POVM $\cm_U^{\rho_0} = \set{\rho_0 + U^\dagger
      P_1 U,\,U^\dagger P_2 U,\,\ldots,\,U^\dagger P_N U}$, with
    $\rho_0 + \sum U^\dagger P_n U= \id$. In appendix~C4, we show that
    \begin{equation}
        \label{eq:35}
        \braket{D_{\cm_U^{\rho_0}}(\rho_t,\omega)}_U
        \leq D_{\rho_0}(\rho_t,\omega) + \frac{1}{2} \sqrt{\frac{N}{d-1} }.
    \end{equation}
\end{corollary}
This means most measurements are already equilibrated even for a large
number of outcomes as long as $N \ll d$, a physically reasonable
assumption for systems composed of many particles given that the
dimension $d$ grows exponentially with the number particles.
Furthermore, for any $N \ll d$, most measurements with a definite
initial value (which are typically out of equilibrium initially) still
equilibrate essentially as fast as a \rank{1} projector.

\section{Slow Equilibration}
\label{sec:slow-equil-res}

% This result complements the previous two sections, for it explicitly
% constructs a high-rank projector $P$ for which the average
% distinguishability $\braket{D_P(\rho_t,\omega)}_T$ does not reduce
% fast. For clarity, we state the result colloquially here and leave
% the theorem and definitions for appendix~\ref{sec:it-stays-inside}.
%
% For any pure system with high effective dimension, it is always
% possible to define a measurement for which the equilibration time is
% tremendously long (typically greater than the age of the universe).
% We do that by defining a subspace such that the wave function
% \emph{(i)} stays predominantly inside it for a long time, but
% nevertheless \emph{(ii)} equilibrates predominantly out of it for
% longer times. In the appendix~\ref{sec:it-stays-inside}, we provide
% an example where this time is $\sim \frac{\de}{1000\sigma_E}$ and
% show this can easily be longer than the age of the universe.
% Meanwhile, for the same system, the typical time scale is on the
% order of femptoseconds.

%%%%%%%%%%%%%%%%%%%%%%%%%%%%%%%%%%%%%%%%%%%%%

This result complements the previous sections by showing that fast
equilibration is not always the case. We find that for any pure system
with high effective dimension it is always possible to define a
measurement for which the equilibration time is tremendously long.
% (typically greater than the age of the universe).

We do that by considering the projector $P_\hk$ onto the subspace
$\hk$ defined by
\begin{equation}
    \label{eq:10}
    \hk = \spn\Set{\ket{\psi(j \tau)}}{j=0,\ldots,K-1},
\end{equation}
with $\tau = 2 \epsilon / \sigma_E$. We prove in appendix~D that, for any $\epsilon$,
% the distinguishability satisfies the following two equations
\begin{equation}
    \label{eq:9}
    \!\!D_{P_\hk}\!(\rho_t,\omega) \geq 1 - \epsilon^2 - \sqrt{\frac{K}{\de} },
    \quad \forall t\! \in \!\left[ 0,K \tau - \frac{\epsilon}{\sigma_E} \right],
\end{equation}
and
\begin{equation}
    \label{eq:20}
    \braket{D_{P_\hk}(\rho_t,\omega) }_{T \rightarrow \infty} \leq 2 \sqrt{\frac{K}{{\de}}} \ll 1,
\end{equation}
showing that the distinguishability is above some constant for a time
that can be very long, but still equilibrates eventually.

The construction simply takes the subspace comprised of $K$ sequential
``snapshots'' of the wave function, and makes sure that the time step
between these snapshots is small enough such that the wave-function
does not move out during the intermediate times. The \emph{``time is
  long''} statement holds because the necessity to take small steps is
nothing compared to the very large number of steps we are allowed to
include ($K\ll \de$). In the appendix~D we provide an example where
this time is $\sim \frac{\de}{1000\sigma_E}$, which can easily be
longer than the age of the universe.

\section{Discussion}
\label{sec:discussion}
\label{sec:harm-osc}
\begin{figure}
    \centering
    \includegraphics[width=0.5\textwidth]{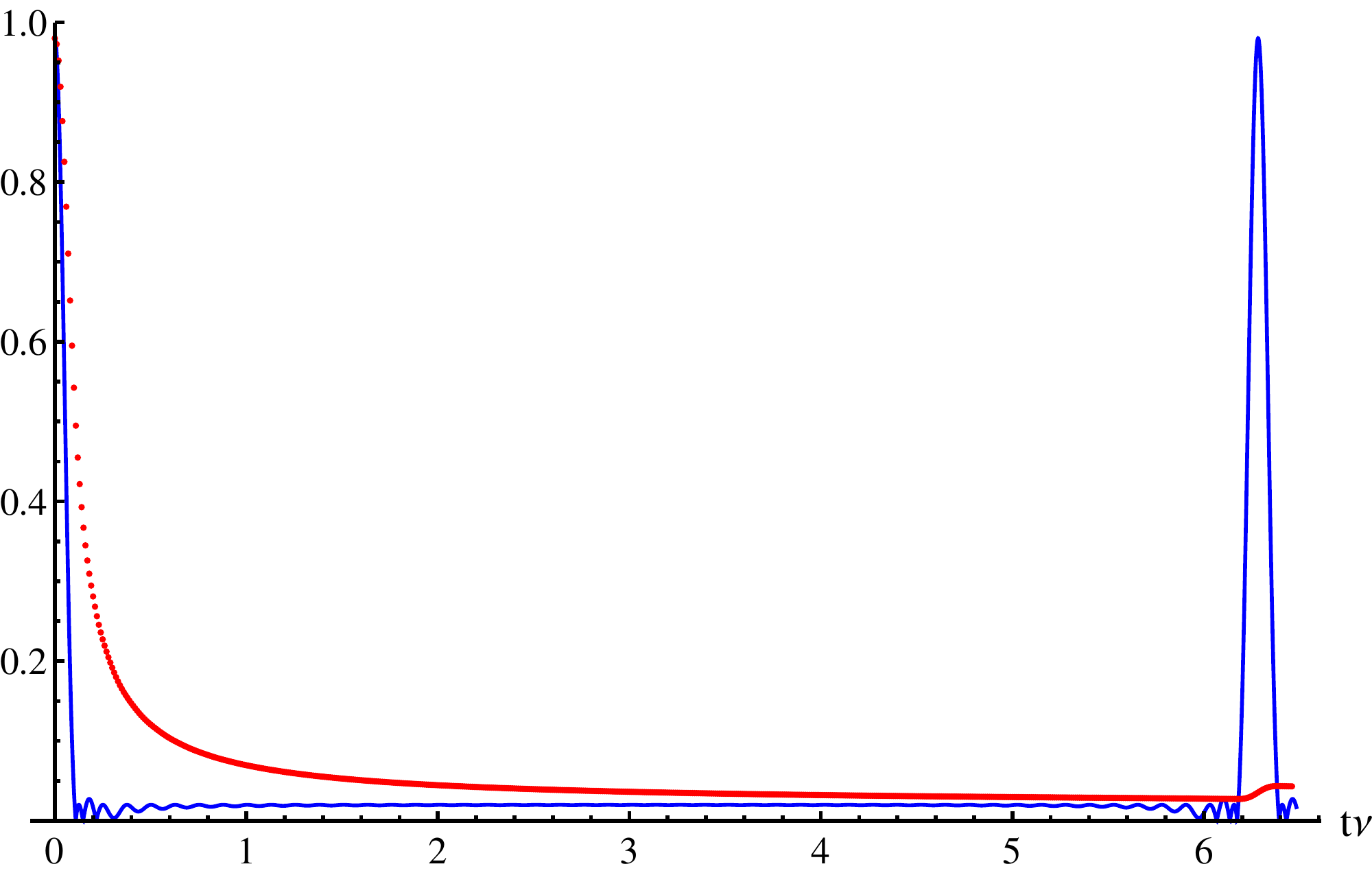}
    \caption{\(D_{\rho_0}(\rho_t,\omega)\) and its finite-time average
      for a full period of the harmonic oscillator with level spacing
      $\nu$ and the initial condition is a pure state spread equally
      over the first 50 energy levels (irrespective of phases). Notice
      that, despite the revival (blue solid line), the projector still
      equilibrates (red dotted line).}
    \label{fig:3}
\end{figure}
In this work we have proved several properties regarding observable
equilibration all of which apply to any system capable of
equilibration.

First we find an upper bound on the time scale of equilibration of
\emph{any} two-outcome measurement based on the rank of the projector
that defines it, which turns out to be very fast for small ranks. We
also find that typical measurements of any rank and any reasonable
number of outcomes are already equilibrated. To investigate time
scales we then turn to a natural class of measurements which are
typically initially out of equilibrium (those for which the initial
state gives a definite value) and show that most of these measurements
equilibrate fast -- approximately as fast as a rank-$1$ projector. On
the other hand, we construct a measurement which is extremely slow to
equilibrate. This shows that indeed, in order to obtain physically
realistic time scales, one must restrict to further constraints on the
measurements and/or on the system considered.

% For a generalization to $N$-outcome measurements of
% sections~\ref{sec:fast-equil-observ} and~\ref{sec:slow-equil-res}
% see appendix~\ref{sec:n-outc-meas}.

% describes the equations, but the take-away message is that
% Theorem~\ref{thr:4} generalizes well when one of the outcomes is of
% very large rank, while theorems~\ref{thr:1}, \ref{thr:3},
% and~\ref{thr:5} always generalize almost identically.

One characteristic that distinguishes the present work from some
previous results~\cite{ShortFarrelly11,Short11,Reimann10,Lin09,Lin10}
is that there was no need to assume non-degenerate energy gaps in
order to prove equilibration. % Indeed, Theorem~\ref{thr:4} applies to
% any Hamiltonian, regardless of the gaps spectrum. Similarly the
% distinguishability for typical projectors decays to zero fast even
% for systems with many degenerate energy gaps, like the harmonic
% oscillator. Even in Theorem~\ref{thr:1} the condition can be dropped
% if one restricts $K\ll \de$ (which can still be extremely large),
% because eq.~(\ref{eq:20}) can also be upper-bound by $2
% \sqrt{K/\de}$.
To emphasize this, Figure~\ref{fig:3} plots an example of the
distinguishability $D_{\rho_0}(\rho_t,\omega)$ of a harmonic
oscillator (with highly degenerate gaps) against its time average. The
function in the Figure decays fast for large $d$, with $\teq \sim
\frac{1}{d \nu}$, which implies that typical projectors equilibrate
fast, as given by equation~(\ref{eq:34}). Nevertheless, the function
returns to its original value at multiple times of $T_{\text{rev}} =
\frac{2\pi}{\nu}$, times at which a full revival manifests. This does
not conflict with equilibration because these revivals are so short
that they cannot affect the average significantly.

The results described here aim to be general, by making statements as
a function of the rank, and which apply to any system. However, there
are specific cases which deserve special attention. When the
measurement is restricted to a small subsystem of a complex many-body
system it is expected to equilibrate fast; however,
Theorem~\ref{thr:4} by itself does not lead to that conclusion since
the outcomes of these measurements are of high rank. Moreover, typical
measurements (in the Haar measure sense) need not necessarily
represent physically relevant measurements. It would be interesting to
study whether these results can be extended to typical measurements
with certain constraints, for instance measurements acting on a small
subsystem.

\acknowledgements We would like to thank Andreas Winter for helpful
discussions. AJS acknowledges support from the Royal Society. ASLM
acknowledges support from the CNPq. Part of this work was done while
the authors were at the Newton Institute programme on Mathematical
Challenges in Quantum Information.

\vspace{0.3cm}

\emph{Note added:} While finalizing this manuscript we became aware of
very recent independent work \cite{Goldsteinnew} which also addresses
the issue of the rapid equilibration of quantum systems.

\bibliography{references}

\appendix
%%%%%%%%%%%%%%%%%%%%%%%%%%%%%%%%%%%%%%%%%%%%%%%%%%%%%%%%%%%%%%%%%%%%%%%%%%%
%%%%%%%%%%%%%%%%%%%%%%%%%%%%%%%%%%%%%%%%%%%%%%%%%%%%%%%%%%%%%%%%%%%%%%%%%%%
%%%%%%% Here begin the appendices
%%%%%%%%%%%%%%%%%%%%%%%%%%%%%%%%%%%%%%%%%%%%%%%%%%%%%%%%%%%%%%%%%%%%%%%%%%%
%%%%%%%%%%%%%%%%%%%%%%%%%%%%%%%%%%%%%%%%%%%%%%%%%%%%%%%%%%%%%%%%%%%%%%%%%%%
\section{General bound on time scales of equilibration}
\label{sec:general-bound}

Following Short and Farrelly~\cite{ShortFarrelly11}, we focus on the
average distance between the expected value of some observable $A$ and
its infinite time average. As proved in a footnote in the
main text~\footnote{Define \unexpanded{$\Theta_T(t) = \frac{1}{T}$}
  for \unexpanded{$t \in [0,T]$}, and \unexpanded{$0$} otherwise.
  Then \unexpanded{$\braket{f}_T = \int_{- \infty}^{\infty}
    \Theta_T(t)f(t) \dif t$}, and \unexpanded{$\Theta_T(t) \leq
    \frac{5}{4} \frac{T}{T^2 + (t-T/2)^2}$}. }, the usual average
can be bounded by the Lorentzian average, therefore
\begin{align}
    \big\langle \left| \tr{A (\rho_t-\omega) } \right|^2 \big\rangle_T
    &\leq \frac{5\pi}{4} \big\langle  \left| \tr{A (\rho_t-\omega)} \right|^2 \big\rangle_{L_T} ,
\end{align}
where $\rho_t$ is the instantaneous state and $\omega = \lim_{T \rightarrow \infty} \braket{\rho_t}_T$ is the infinite time-averaged state.

Denoting by $E_j$ and $\ket{j}$ the eigenvalues and eigenvectors of the
Hamiltonian, and assuming an initially pure state for simplicity, the evolved state is
\begin{equation}
    \rho_t = \sum_{j,k=1}^{\tilde{d}} c_j c_k^* e^{-i(E_j - E_k)t} \ket{j} \bra{k}.
\end{equation}
Defining the matrix elements $\bra{j} A \ket{k} = A_{jk}$ one has
\begin{align}
 \langle  \left| \tr{A (\rho_t-\omega)} \right|^2 \rangle_{L_T}
    &= \Big\langle \Big| \sum_{j \ne k } (c_k^*A_{kj} c_j ) e^{-i(E_j - E_k)t} \Big|^2 \Big\rangle_{L_T} \notag\\
    &= \sum_{j \ne k, n \ne l} (c_k^* A_{jk} c_j ) (c_l^* A_{ln} c_n )^* \notag\\
    &\qquad\qquad\times\big\langle e^{-i[(E_j - E_k) - (E_n - E_l)]t} \big\rangle_{L_T}.
\end{align}

Each energy gap can be labeled by $G_{(j,k)} = E_j - E_k$ with indexes
$\alpha = (n,l)$ and $\beta = (j,k)$. In this way we define a vector
$v$ and a Hermitian matrix $M$
\begin{equation}
    v_{\alpha} = v_{n,l} = c_l^* A_{ln} c_n, \qquad M_{\alpha \beta}= \big\langle e^{i(G_\alpha - G_\beta)t} \big\rangle_{L_T}.
\end{equation}

With the above definitions we can see
\begin{align}
    & \langle  \left| \tr{A (\rho_t-\omega)} \right|^2 \rangle_{L_T}
    = \sum_{\alpha \beta} v_\alpha^* M_{\alpha \beta} v_\beta \le \|M\| \sum_{\alpha} |v_\alpha|^2 \nonumber \\
    &\quad\quad\le \|M\| \sum_{i ,j} |c_i|^2 |c_j|^2 |A_{ji}|^2 = \|M\| \text{Tr} (A\omega A^\dag\omega) \nonumber \\
    &\quad\quad\le \|M\| \sqrt{\text{Tr}(A^\dag A \omega^2)\text{Tr}(AA^\dag \omega^2)} \nonumber \\
    &\quad\quad\le \|M\| \|A\|^2 \text{Tr}(\omega^2) = \|M\| \frac{\|A\|^2}{d_{\text{eff}}}. \label{coditiontony}
\end{align}
On the first step, from the definition of the spectral norm
$\|M\|.\|v\| \ge \|Mv\|$ so $\|v\|.\|M\|.\|v\| \ge \|v\|. \|Mv\| \ge
|v^\dag M v|$. The other steps come from using the Frobenius inner
product and Cauchy-Schwarz inequality, and from the fact that for pure states $\tr{\omega^2}
= 1/\de$, with $\de = \Big( \sum_j |c_j|^2 \Big)^{-1}$.

By use of the identity
$\langle e^{i \nu t} \rangle_{L_T} = e^{-|\nu|
  T} \expo{i\nu T/2}$ we have $\left| M_{\alpha \beta} \right| = \left| \big\langle e^{i(G_\alpha -
    G_\beta)t} \big\rangle_{L_T} \right| = e^{-|G_\alpha - G_\beta|T}$
and,
since $M$ is a Hermitian matrix, standard results give
\begin{equation}
    \|M\| \le \max_\beta \sum_\alpha |M_{\alpha \beta}| = \max_\beta \sum_\alpha e^{-|G_\alpha - G_\beta|T}.
\end{equation}
We can now break the sum into intervals of width $\epsilon$, centered
around a given gap $G_\beta$. An interval $\epsilon$ can fit at most
$\mathcal{N}(\epsilon)$ gaps which satisfy $(k+1/2)\epsilon > G_\alpha
- G_\beta > (k-1/2)\epsilon$, which in turn implies $|G_\alpha -
G_\beta| \ge (|k| - 1/2)\epsilon$. Therefore
\begin{equation}
    |M_{\alpha \beta}| \le e^{-(|k| - 1/2) \epsilon T}.
\end{equation}
For the case $k = 0$ we just use the fact that $|M_{\alpha \beta}| \le
1$.

The sum is maximized by taking as many small values of $|k|$ as
possible, and since there are $\tilde{d}(\tilde{d} - 1)$ terms in total we have that
\begin{align}
    \max_\beta
    &\sum_\alpha |M_{\alpha \beta}| \le \mathcal{N}(\epsilon)
    \Bigg( 1 + 2 \sum_{k=1}^{\tilde{d}(\tilde{d} - 1)} e^{-(k-1/2)\epsilon T } \Bigg) \nonumber \\
    &= \mathcal{N}(\epsilon) \Bigg( 1 + 2 e^{\epsilon T/2 }
    \frac{e^{- \epsilon T }\big(e^{- \epsilon T \tilde{d}(\tilde{d} - 1)} - 1 \big) }{e^{- \epsilon T } - 1} \Bigg) \nonumber \\
    &\le \mathcal{N}(\epsilon) \Bigg( 1 + 2 \frac{e^{- \epsilon T /2}}{1 - e^{- \epsilon T }} \Bigg).
\end{align}

Finally, by using $\frac{1}{1-e^{-x}} \le 1 + \frac{1}{x}$ we get
\begin{align}
    \langle \left| \tr{A (\rho_t-\omega) } \right|^2 \rangle_T\!
    &\le \! \frac{5\pi}{2} \frac{\|A\|^2}{d_{\text{eff}}}
    \mathcal{N}(\epsilon) \bigg( \frac{1}{2} +\! e^{- \epsilon T /2}
    +\! \frac{e^{- \epsilon T /2}}{\epsilon T } \bigg) \notag\\
    &\le \frac{5\pi}{2} \frac{\|A\|^2}{d_{\text{eff}}}
    \mathcal{N}(\epsilon) \bigg( \frac{3}{2} + \frac{1}{\epsilon T } \bigg).
    \label{bound1}
\end{align}
Comparing this expression to the result in~\cite{ShortFarrelly11},
there is an improvement in the bound of order $\sim \log_2(\tilde{d})$.

The result is taken in the original paper by Short and Farrelly as a
stepping stone to obtain bounds on the time scale of equilibration
with respect to the distinguishability. By the same procedure they
take, we obtain that the average distinguishability for a set of
measurements $\mathcal{M}$ satisfies
\begin{equation}
    \big\langle D_{\mathcal{M}}(\rho_t,\omega) \big\rangle_T
    \le \frac{\mathcal{S}(\mathcal{M})}{4}
    \sqrt{\frac{5\pi \mathcal{N}(\epsilon)}{2 d_{\text{eff}}}
      \bigg( \frac{3}{2} + \frac{1}{\epsilon T } \bigg) }
\end{equation}
where $\mathcal{S}(\mathcal{M})$ is the total number of outcomes of
all the possible measurements.

A simple estimate illustrates how long these bounds on the time scales
still are. If one assumes the energy levels are more or less equally
distributed, the minimum distance between energy gaps scales as
$\epsilon_{\text{min}} \leq \frac{\Delta U}{\tilde{d}^2} $, with $\Delta U$
being the total energy range. This gives an equilibration time that
scales very roughly as $T_{\text{eq}} \sim \frac{1}{\de \epsilon} >
\frac{\tilde{d}}{\Delta U}$, which is terribly long for systems composed of
more than a few particles.

\section{Theorem~\ref{thr:4}}
\label{sec:theorem-refthr:4}
\label{sec:fast-equil-observ-1}

For a general initial state given by $\rho =
\sum_{jk} \rho_{jk} \ket{j} \bra{k}$, the purity of $\omt$ can be written
\begin{align}
    \tr{\omega_{L_{T}}^2} &= \tr{\omt\omt^\dagger}\notag\\
    &= \operatorname{Tr} \left[\sum_{n,m}\rho_{nm}\left<e^{-i(E_n-E_m)t}\right>_{L_{T}}\ket{n}\bra{m}\right.\notag\\[-14pt]
    &\qquad\qquad\qquad\left.\sum_{j,k}\rho^*_{jk}\left<e^{i(E_j-E_k)t}\right>_{L_{T}}\ket{k}\bra{j} \right]\notag\\
    &= \sum_{j,k}|\rho_{jk}|^2 \left| \left<e^{-i(E_j-E_k)t}\right>_{L_{T}} \right|^2 \notag\\
    &\leq \sum_{j,k}\rho_{jj}\rho_{kk} \left| \left<e^{-i(E_j-E_k)t}\right>_{L_{T}} \right|^2 \notag\\
    &= \sum_{j,k}p_{j}p_{k} \left| \left<e^{-i(E_j-E_k)t}\right>_{L_{T}} \right|^2,
    \label{eq:2}
\end{align}
where the previous to the last line is an equality for an initially
pure state, and the inequality follows in general from positivity of
the density operator~\footnote{\unexpanded{$\bra{v}\rho\ket{v}\geq 0$}
  for all \unexpanded{$\ket{v}$}, which applies in particular to
  \unexpanded{$\ket{v}=a\ket{j}+b\ket{k}$},
  so \unexpanded{$\left(\begin{smallmatrix} \rho_{jj}&\rho_{jk}\\
          \rho_{kj}&\rho_{kk}\end{smallmatrix}\right)$} is positive
  too, and, since the determinant must be greater than or equal to
  zero, \unexpanded{$|\rho_{jk}|^2\leq \rho_{jj}\rho_{kk}$}.}.

By use of the identity $\langle e^{i \nu t} \rangle_{L_T} = e^{-|\nu|
  T} \expo{i\nu T/2}$ we can in turn see
\begin{align}
    \label{eq:41}
    \tr{\omega_{L_{T}}^2} &\le \sum_{jk} p_j p_k \expo{-2\left| E_j - E_k \right|T}
\end{align}
the above being an equality for pure states.

To see the connection to $\eta_{\frac{1}{T}}$, we define the function
\begin{equation}
    \label{eq:16}
    g(x) =
    \begin{cases}
        1, \mbox{ if } x \in [0,1) \\
        0, \mbox{ otherwise. }
    \end{cases}
\end{equation}
This definition is important because it allows us to upper bound the
exponential as
\begin{equation}
    \label{eq:3}
    \expo{-\left| x \right|}\leq \sum_{n = 0}^{\infty} \expo{-n \delta} g \left( \frac{\left| x \right|}{\delta} - n \right),
    \quad \forall \delta > 0.
\end{equation}
So we have
\begin{align}
    \label{eq:4}
    \tr{\omt^2} &\leq \sum_{n = 0}^{\infty} \expo{-n \delta} \sum_{j} p_j \sum_{k} p_k \notag\\
    &\quad\quad\qquad\times g \left( \frac{2\left| E_j - E_k \right|T}{\delta} - n \right) \notag\\
    &= \sum_{n = 0}^{\infty} \expo{-n \delta} \sum_{j} p_j \!\! \!\!
    \sum_{\substack{k\;: \\ \left(2\left| E_j - E_k \right| \frac{T}{\delta} - n \right) \in [0, 1)}} \!\! \!\! \!\! \!\! p_k \notag\\
    & \leq \sum_{n = 0}^{\infty} \expo{-n \delta} \sum_{j} p_j
    \left[\sum_{\substack{k\;:\\ E_k \in I_{-}}} p_k + \sum_{\substack{k\;:\\ E_k \in I_{+}}} p_k \right]
    \notag\\
    &\leq \sum_{n = 0}^{\infty} \expo{-n \delta} \sum_{j} p_j \left(2 \eta_{\frac{\delta}{2T}} \right) \notag\\
    &= \frac{2 \eta_{\frac{\delta}{2T}}}{1-\expo{-\delta}}
\end{align}
where $I_{+} = \left[E_{+}, E_{+}+ \frac{\delta}{2T} \right)$,
$I_{-} = \left(E_{-} - \frac{\delta}{2T}, E_{-}\right]$,
$E_\pm=E_j \pm \frac{n \delta}{2T}$, and the inequality in the
penultimate line applies for any combination of $n$ and $j$.

The important fact is that $\delta$ is ours to define, and determines
the balance between the quotient and $\eta$. We can, for instance, set
it to $2$
\begin{equation}
    \label{eq:5}
    \tr{\omt^2} \leq \frac{2\eta_{\frac{1}{T} }}{1-\expo{-2}} < 2.32 \eta_{\frac{1}{T} }
\end{equation}
and we know the system has equilibrated when the total probability
inside any energy interval of size $1/T$ is small. We can also
manually fix the energy interval with $\delta =2T \Delta E$, so
\begin{align}
    \label{eq:6}
    \tr{\omt^2} &\leq \frac{2\eta_{\Delta E}}{1 - \expo{-2 T \Delta E}}
\end{align}
which still leaves us with a free variable, but clearly singles out
the time-dependence.

In any case, we have that
\begin{align}
    \label{eq:7}
    \braket{\tr{\rho_t P}}_T
    &\leq \frac{5\pi}{4} \sqrt{K\, \frac{2\eta_{\frac{1}{T} }}{1-\expo{-2}}} \notag\\
    &< 6 \sqrt{\eta_{\frac{1}{T} } K }.
\end{align}

 For clarity, Figure~\ref{fig:1}
 displays a graphic illustration of $\eta_\epsilon$.
 \begin{figure}
     \centering
     \includegraphics[width=0.5\textwidth]{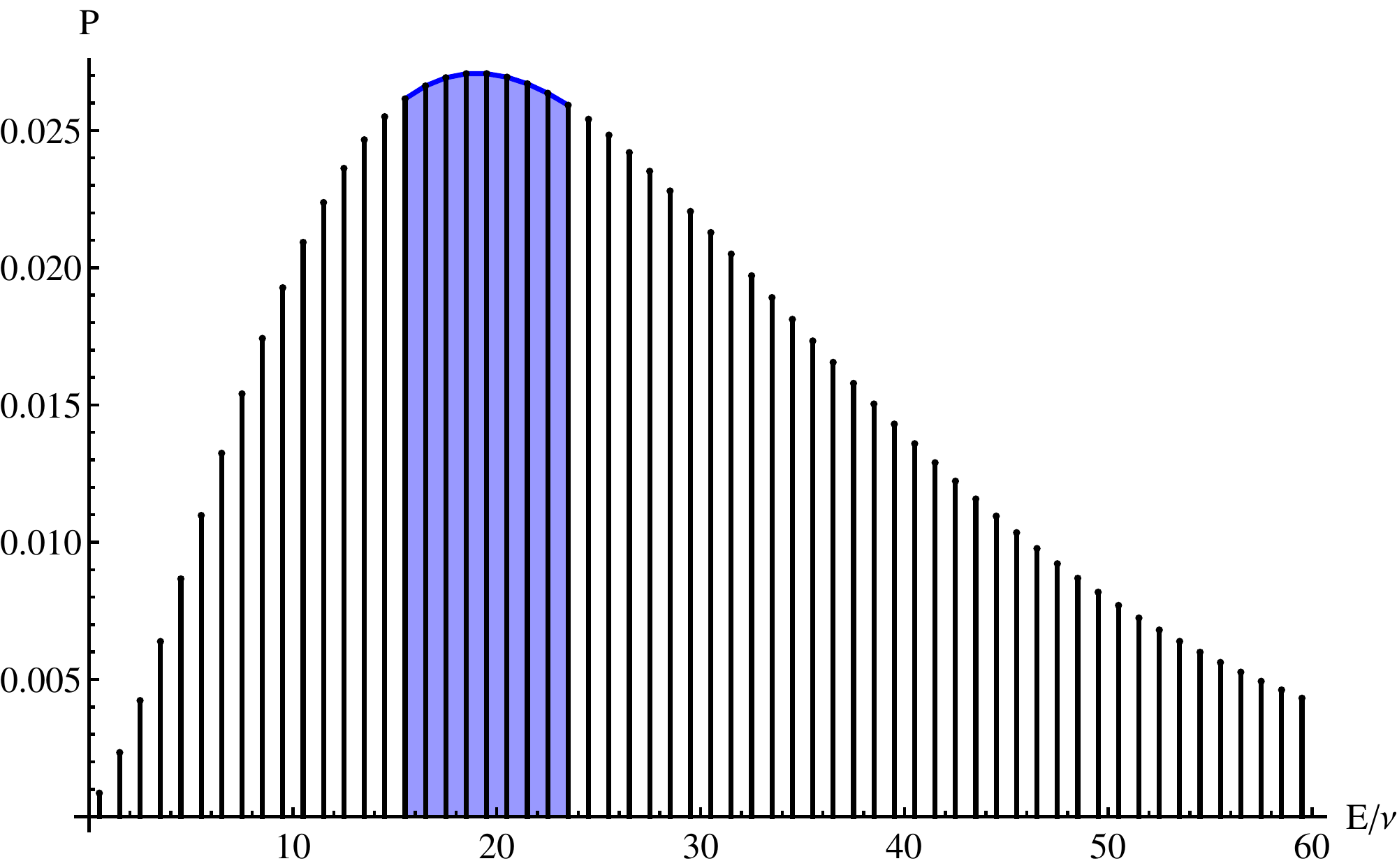}
     \caption{Graphic illustration of $\eta_\epsilon$. The vertical
       lines display the probability distribution in energy space of a
       $3$-dimensional harmonic oscillator with energy levels $E_n =
       (n+1/2)\nu$, under the Boltzmann distribution (accounting for
       degeneracies) with temperature $= 10\nu$. The blue region
       represents $\eta_{8\nu}$ (the maximum probability that can be
       found inside an energy interval of size $8\nu$).}
     \label{fig:1}
 \end{figure}

\subsection{Gaussian Distribution Example}
\label{sec:gauss-distr}

Taking a pure initial state whose probability
distribution can be approximated by
the following continuous function of energy
\begin{equation}
    \label{eq:22}
    p(E) = \frac{1}{\sqrt{2\pi} \sigma_E} e^{-\frac{E^2}{2\sigma_E^2}},
\end{equation}
we show below that the purity of
$\omt$
can be approximated by
\begin{equation}
    \label{eq:23}
    \tr{ \omega^2_{L_T} } \approx \frac{1}{2\sqrt{\pi} \sigma_E T},
\end{equation}
for $\sigma_E T \gg 1$.
Meanwhile, we also show that $\eta$ satisfies
\begin{equation}
    \label{eq:24}
    \frac{2\eta_{\frac{1}{T} }}{1-\expo{-2}} < \frac{3.28}{2\sqrt{\pi}\sigma_E T},
\end{equation}
and therefore
\begin{equation}
    \label{eq:8}
        {\eta_{\frac{1}{T} }} < \frac{0.4}{\sigma_E T}.
\end{equation}
It is interesting to see that, despite the approximations taken, the
function $\eta$ gives good estimates for the purity of the average,
with much simpler calculations.

\subsubsection{Calculations}
\label{sec:purity-decay-timesc-1}
The Fourier transform of the energy distribution defined in eq.~\eqref{eq:22}
is
\begin{equation}
    \mu(t) = \int_{-\infty}^{\infty} p(E) e^{-i E t} dE
    = e^{-\frac{\sigma_E^2 t^2}{2}}.
\end{equation}
Thus, using the continuous version of eq.~(\ref{eq:41}) and
assuming a pure state,
\begin{align}
    \tr{\omega_{L_{T}}^2} &= \iint_{- \infty}^{\infty} dE\, dE'\, p(E) p(E') \expo{-2\left| E - E' \right|T} \notag\\
    &= \iint dE\, dE'\, p(E) p(E') \int_{- \infty}^{\infty} \frac{dt }{\pi} \frac{\expo{i(E - E')2t}T}{T^2+t^2}  \notag\\
    &= \int_{- \infty}^{\infty} \frac{dt }{\pi} \frac{ T}{T^2+t^2} \left| \mu(2t) \right|^2  \notag\\
    &= \frac{T}{\pi} \int_{-\infty}^{\infty}
    \frac{e^{-4 \sigma_E^2 t^2}}{T^2 + t^2} dt \notag\\
    &= \frac{1}{\pi} \int_{-\infty}^{\infty}
    \frac{e^{-4 \sigma_E^2 T^2 x^2}}{1 + x^2} dx \notag\\
    &= e^{4 \sigma_E^2 T^2 } \left( 1 - \erf\left(2\sigma_E T\right) \right),
\end{align}
where the second line uses the inverse Fourier transform of $\expo{-2\left| E - E' \right|T}$.

We can use the expansion of the
Error Function
\begin{equation}
    \erf(x) = 1 - \frac{e^{-x^2}}{\sqrt\pi x} + e^{-x^2}\mathcal{O}\left( x^{-3} \right)
\end{equation}
which, for $\sigma_E T$ large enough, results in eq.~\eqref{eq:23}.

\subsubsection{Comparison to $\eta$}
\label{sec:comparison-}

As we have established previously
\begin{equation}
    % \label{eq:1}
    \tr{\omt^2} < 2.32 \eta_{\frac{1}{T} }.
    \tag{\ref{eq:5}}
\end{equation}
Since the density $p(E)$ is a Gaussian, $\eta_\epsilon$ is obviously
at its center. Thus
\begin{align}
    \label{eq:2-2}
    &\eta_{\frac{1}{T} } = \int_{-\frac{1}{2T} }^{\frac{1}{2T} }
    \frac{1}{\sqrt{2\pi}\sigma} \expo{-\frac{E^2}{2\sigma^2} } \dif E
    \leq \frac{1}{\sqrt{2\pi}\sigma T} \notag\\
    &\tr{\omt^2} < \frac{3.28}{2\sqrt{\pi}\sigma T},
\end{align}
where the integral was trivially approximated by
$\int_{0}^{\epsilon} p(E) \dif E < \epsilon p(0)$.

\subsubsection{Infinite Time Limit}
\label{sec:an-extra-point}
The upper bounds we calculated in the example for
$\tr{\omega^2_{L_T}}$ tend to zero as $T$ tends to infinity, which
seems to contradict the fact that for the infinite time average state
$\omega$, $\tr{\omega^2}$ is not zero. The reason this occurs is that
we need to be careful when averaging expressions like $\sum
p_j p_k\exp[-2i(E_j-E_k)t]$ over energy levels. The terms with
$E_j=E_k$ in the integrals arising from averages over $p(E)$ do not
contribute (they are of measure zero), whereas in the finite sum they
did, giving $\sum_{j=k}p_j p_k=1/\de$.

Another way of looking at this is that taking the continuous limit
implies taking $\de = \infty$.

\section{Typical Projectors}
\label{sec-1}

Here, we prove theorems~\ref{thr:3} and~\ref{thr:5}.

\subsection{Proof of Theorem~\ref{thr:3}}
\label{sec-1-1}\label{sec:already-equil}
Here, we use the fact (well known from representation theory) that,
for any operator $M$,
\begin{equation}
    \label{eq:37}
    \braket{U^\ts M(U^\ts)^\dagger}_U =
    % \frac{\tr{\Pi_S M}}{\tr{\Pi_S}}\Pi_S + \frac{\tr{\Pi_A M}}{\tr{\Pi_A}} \Pi_A,
    \alpha\Pi_S + \beta \Pi_A,
\end{equation}
where $\Pi_S = (\id^\ts + S)/2$ and $\Pi_A = (\id^\ts - S)/2$ are the
projectors onto the symmetric and antisymmetric subspaces of
dimensions $\frac{d(d+1)}{2} $ and $\frac{d(d-1)}{2} $ respectively,
and $S$ is the swap operator on $\Hil\otimes\Hil$: $S \ket{a}\ket{b} =
\ket{b}\ket{a}$. Given that $\tr{SA\otimes B} = \tr{AB}$,
it is straightforward to see that, for $M=P^\ts$,
\begin{align}
    \label{eq:230803}
    \alpha &= \frac{\tr{\Pi_S P^\ts}}{\tr{\Pi_S}}  \notag\\
    &= \frac{1}{d(d+1)} \left[ \tr{P\otimes P}+\tr{P^2} \right] \notag\\
    &= \frac{K(K+1)}{d(d+1)} ;
\end{align}
and similarly $\beta = \frac{\tr{\Pi_A P_U^\ts}}{\tr{\Pi_A}} = \frac{K(K-1)}{d(d-1)}$.

Finally, going back to the distinguishability and using the fact that
$\tr{\rho_t\omega}=\tr{\omega^2}$,
\begin{align}
    \label{eq:135185}
    &\braket{D_{P_U}(\rho_t,\omega)^2}_U \notag\\
    &\quad= \braket{\left( \tr{P_U(\rho_t-\omega)} \right)^2}_U \notag\\
    &\quad= \tr{\braket{P_U\otimes P_U}_U(\rho_t - \omega)^\ts}. \notag\\
    &\quad= \tr{\left( \alpha\Pi_S + \beta\Pi_A \right)
      (\rho_t - \omega)\otimes(\rho_t - \omega)} \notag\\
    &\quad= \frac{\alpha+\beta}{2}\left( \tr{\rho_t-\omega} \right)^2
    + \frac{\alpha-\beta}{2}\tr{(\rho_t-\omega)^2} \notag\\
    &\quad= \frac{K}{2d} \left( \frac{K+1}{d+1} - \frac{K-1}{d-1} \right)
    \left( \tr{\rho_t^2}+\tr{\omega^2}-2\tr{\rho_t\omega}  \right)\notag\\
    &\quad= \frac{K}{d} \frac{d-K}{d^2-1}\tr{\rho_t^2-\omega^2}.
\end{align}

\subsection{Proof of Theorem~\ref{thr:5}}
\label{sec-1-2}
\label{4aead996-94b1-4025-9891-e7bf9aa7c756}
\label{sec:average-proj-cont-rho}

To simplify notation in the proof, we will assume $d>2$ and set $d' = \dim{\Hil'}
= d-1$ and $ K' = \rank{P_U}  = K-1$.

Again, the quantity we are interested in will be
\begin{align}
    \label{eq:831588}
    \braket{D_{\Pi_U}(\rho_t,\omega)}_U &= \braket{\left| \tr{\Pi_U(\rho_t-\omega)} \right|}_U \notag\\
    &\leq \sqrt{\braket{\Big( \tr{\Pi_U(\rho_t-\omega)} \Big)^2 }_U} \notag\\
    &= \sqrt{\tr{\braket{\Pi_U\otimes \Pi_U}_U(\rho_t - \omega)^\ts}}.
\end{align}
Using the results from section~\ref{sec:already-equil}, it is easy to
see that
\begin{align}
    \label{eq:506655}
    &\braket{\Pi_U\otimes \Pi_U}_U \notag\\
    &= \braket{\rho_0\otimes \rho_0 + \rho_0\otimes P_U + P_U \otimes \rho_0 + P_U \otimes P_U}_U \notag\\
    &= \rho_0\otimes \rho_0 + \rho_0\otimes \braket{P_U}_U \notag\\
    &\qquad+ \braket{P_U}_U \otimes \rho_0 + \braket{P_U \otimes P_U}_U \notag\\
    &= \rho_0\otimes \rho_0 + \frac{K'}{d'} \rho_0\otimes \id' + \frac{K'}{d'} \id' \otimes \rho_0 \notag\\
    & \qquad + \frac{K'(K'+1)}{d'(d'+1)}\Pi_S' + \frac{K'(K'-1)}{d'(d'-1)}\Pi_A',
\end{align}
where $\id' = \id - \rho_0$, $\Pi_S' = \id'^\ts\Pi_S\id'^\ts$ and $\Pi_A' =
\id'^\ts\Pi_A\id'^\ts$. So that
\begin{align}
    \label{eq:245099}
    &\braket{\tr{\Pi_U(\rho_t-\omega)}^2}_U \notag\\
    &\qquad= \tr{\braket{\Pi_U\otimes \Pi_U}_U(\rho_t - \omega)^\ts} \notag\\
    &\qquad= f(t)^2 - 2f(t)^2 \frac{K'}{d'} + \frac{K'(K'+1)}{d'(d'+1)}\tr{\Pi_S A(t)} \notag\\
    &\qquad\qquad+ \frac{K'(K'-1)}{d'(d'-1)} \tr{\Pi_A A(t)},
\end{align}
where $A(t) = \left[ (\id - \rho_0)(\rho_t - \omega)(\id - \rho_0) \right]^\ts$, and
$f(t) = \tr{\rho_0(\rho_t - \omega)}$.
Since
\begin{align}
    \label{eq:905164}
    \tr{\Pi_S A(t)}
    &= \tr{\frac{\id^\ts + S}{2} \left[ (\id - \rho_0)(\rho_t - \omega)(\id - \rho_0) \right]^\ts} \notag\\
    &= \frac{1}{2} \tr{(\id - \rho_0)(\rho_t - \omega)(\id - \rho_0)(\rho_t - \omega)} \notag\\
    &\qquad+ \frac{1}{2} \Big( \tr{\rho_t - \omega - \rho_0(\rho_t - \omega) } \Big)^2 \notag\\
    &\leq \frac{1}{2} \left( 1-\frac{1}{\de}\right) + \frac{1}{2} \left(f(t) \right)^2,
\end{align}
where the last line is due to $\tr{\Pi X\Pi X} \leq \tr{X^2}$ for a
projector $\Pi$ (in this case, $\id-\rho_0$), which
follows from the Cauchy-Schwarz inequality
$\tr{\Pi X\Pi X} \leq \sqrt{\tr{X^2}\tr{\Pi X\Pi X}} $. Thus, since
\begin{align}
    \label{eq:905608}
    \tr{\Pi_A A(t)}
    &= \tr{\frac{\id^\ts - S}{2} \left[ (\id - \rho_0)(\rho_t - \omega) \right]^\ts} \notag\\
    &\leq \frac{1}{2} \left( f(t)^2 - 1 + \frac{1}{\de}\right),
\end{align}
we have
\begin{align}
    \label{eq:471969}
    &\braket{\tr{\Pi_U(\rho_t-\omega)}^2}_U \notag\\
    &\quad\leq f(t)^2 \left( 1 + \frac{1}{2}\frac{K'(K'+1)}{d'(d'+1)}
      + \frac{1}{2}\frac{K'(K'-1)}{d'(d'-1)} - 2\frac{K'}{d'}\right) \notag\\
    &\quad\qquad+ \frac{1}{2} \left( 1-\frac{1}{\de} \right)\left[ \frac{K'(K'+1)}{d'(d'+1)}
      - \frac{K'(K'-1)}{d'(d'-1)} \right] \notag\\
    &\quad= f(t)^2 \left( 1 -2 \frac{K'}{d'} + \frac{K'}{d'} \frac{K'd'-1}{d'^2-1} \right) \notag\\
    &\quad\qquad+ \left( 1-\frac{1}{\de} \right)\frac{K'}{d'} \frac{d'-K'}{d'^2-1} \notag\\
    &\quad\leq f(t)^2 + \frac{1}{4d'}.
\end{align}
The last line can be derived from the fact that the first
parentheses in the penultimate equation is maximized by $K'=0$, and the
second term is maximized by $K' = d'/2$, along with $\de \leq d$ and
$ \frac{d'^2}{(d'-1)(d'+1)^2}\leq \frac{1}{d'} $.

\subsection{Proof of typical initial distinguishability} \label{sec:initial-dist}

To show that observables with a definite initial value are typically out of equilibrium (and thus undergo a non-trivial equilibration process) we consider the initial distinguishability between $\rho_0$ and  $\omega$ for a measurement of $\Pi_U$, averaged over $U$. As before, we will  set $d' = \dim{\Hil'}= d-1$ and $ K' = \rank{P_U}  = K-1$.
\begin{align}
    \braket{D_{\Pi_U}(\rho_0,\omega)}_U &= \braket{\left| \tr{\Pi_U(\rho_0-\omega)} \right|}_U \notag\\
    &=\braket{\left(1-\tr{\Pi_U \omega} \right) }_U \notag\\
    &=1-\tr{\left[ \rho_0 + \frac{K'}{d'}(\id - \rho_0) \right] \omega} \notag\\
    &=\left(1- \frac{K'}{d'} \right) \left(1 - \tr{\rho_0 \omega} \right) \notag\\
    &\geq \left(1- \frac{K-1}{d-1} \right) \left(1- \frac{1}{\de} \right),
\end{align}
where the last line is an equality if $\rho_0$ is pure.

Note that because refining a measurement (by splitting one outcome into many) can only increase the distinguishability, it follows that
\begin{equation}
    \braket{D_{\cm_U^{\rho_0}}(\rho_t,\omega)}_U \geq \left(1- \frac{K-1}{d-1} \right) \left(1- \frac{1}{\de} \right)
\end{equation}
where here $K$ is the rank of the measurement projector containing $\rho_0$.

\subsection{Proof of Corollary~\ref{co:2}}
\label{sec:corollary-refco:1}

Denoting by $K_j$ the rank $P_j$, we have that
\begin{equation}
    \label{eq:42}
    \sum_j K_j = d'
\end{equation}
and
\begin{align}
    \label{eq:35}
    &\braket{D_{\cm_U^{\rho_0}}(\rho_t,\omega)}_U \\
    &\quad= \frac{1}{2} \braket{D_{\rho_0+P_{1U}}(\rho_t,\omega)}_U
    + \frac{1}{2} \sum_{j=2}^{N} \braket{D_{P_{jU}}(\rho_t,\omega)}_U \notag\\
    &\quad\leq \frac{1}{2} \sqrt{\braket{D_{\rho_0+P_{1U}}(\rho_t,\omega)^2}_U}
    + \frac{1}{2} \sum_{j=2}^{N} \sqrt{\braket{D_{P_{jU}}(\rho_t,\omega)^2}_U}.\notag
\end{align}

Following the proof in appendix~\ref{sec-1-2} above, and using the
fact that
$1-\frac{1}{\de} \leq 1-\frac{1}{d} = \frac{d'}{d'+1}$, leads to
% TODO Too Fast!! (780910)
\begin{align}
    \label{eq:40}
    &\braket{D_{P_{jU}}(\rho_t,\omega)^2}_U \\
    &\qquad \le f(t)^2 \frac{K_j}{d'} \frac{K_j d'-1}{d'^2-1}
    + \frac{K_j}{d'+1} \frac{d'-K_j}{d'^2-1} .\notag
\end{align}
and
\begin{align}
    \label{eq:39}
    &\braket{D_{\rho_0+P_{1U}}(\rho_t,\omega)^2}_U \\
    &\qquad \le f(t)^2 \left( 1 -2 \frac{K_1}{d'} + \frac{K_1}{d'} \frac{K_1 d'-1}{d'^2-1} \right)
    + \frac{K_1}{d'+1} \frac{d'-K_1}{d'^2-1} \notag\\
    &\qquad\leq f(t)^2 \left( 1 + \frac{K_1}{d'} \frac{K_1 d'-1}{d'^2-1} \right)
    + \frac{K_1}{d'+1} \frac{d'-K_1}{d'^2-1} \notag\\
    &\qquad= f(t)^2 + \braket{D_{P_{1U}}(\rho_t,\omega)^2}_U.
\end{align}

By using the fact that $\sqrt{a+b}\leq\sqrt{a}+\sqrt{b}$ for $a,b\geq
0$, this leads to
\begin{align}
    \label{eq:45}
    &\braket{D_{\cm_U^{\rho_0}}(\rho_t,\omega)}_U \\
    &\quad\leq \frac{1}{2} \left| f(t) \right|
    + \frac{1}{2} \sum_{j=1}^{N} \sqrt{\braket{D_{P_{jU}}(\rho_t,\omega)^2}_U}.\notag
\end{align}

Through the method of Lagrange multiplier it is easy to see that the
sum in eq.~(\ref{eq:45}), expressed in terms of the $K_j$s through
eq.~(\ref{eq:40}) and constrained by eq.~(\ref{eq:42}), is maximized
by taking all $P_j$ to be of equal rank. This rank must then be $K_j =
d'/N$. Substituting that into eq.~(\ref{eq:40}), and using the inequalities
$\frac{d'^2 - N}{d'^2 - 1}  < 1$,
$1- 1/N < 1$ and $d'^3 \leq (d'+1)^2(d'-1)$,
\begin{align}
    \label{eq:45-2}
    &\braket{D_{\cm_U^{\rho_0}}(\rho_t,\omega)}_U \notag\\
    &\quad\leq \frac{1}{2} \left| f(t) \right|
    + \frac{1}{2} N \sqrt{ \frac{f(t)^2}{N^2} \frac{d'^2-N}{d'^2-1}
      + \frac{d'/N}{d'+1} \frac{d'\left( 1-\frac{1}{N}  \right)}{d'^2-1}}.\notag \\
    &\quad\leq \frac{1}{2} \left| f(t) \right|
    + \frac{1}{2} \sqrt{{f(t)^2} + \frac{N d'^2}{(d'+1)^2(d'-1)}}.\notag\\
    &\quad\leq \left| f(t) \right| + \frac{1}{2} \sqrt{ \frac{N}{d'}}.
\end{align}

\section{Slow Equilibration}
\label{sec:it-stays-inside}

The \emph{Slow Equilibration} result can be rigorously stated as the
following theorem.
\begin{theorem}[Slow equilibration]
    \label{thr:1}
    Given any Hamiltonian, any pure state
    $\ket{\psi(t)}\in \hil$ with effective dimension $\de$, any positive integer $K \ll \de$
    and any $\epsilon > 0$; take $\sigma_E$ to be the standard
    deviation in energy of $\ket{\psi}$, and $P_\hk$ to be the
    projector onto the subspace
    \begin{equation}
        \label{eq:10}
        \hk = \spn\Set{\ket{\psi(j \tau)}}{j=0,\ldots,K-1},
    \end{equation}
    with $\tau = 2 \epsilon / \sigma_E$; then the distinguishability
    satisfies the following two equations~\footnote{The time range in
      equation~(\ref{eq:9}) can be increased to $2K\epsilon/\sigma_E$ by
      a slightly more complicated construction of $\hk$.}
    \begin{equation}
        \label{eq:9}
       \! D_{P_\hk}(\rho_t,\omega) \geq 1 - \epsilon^2 - \sqrt{\frac{K}{\de} },
        \quad \forall t\! \in \!\left[ 0,\frac{(2K-1)\epsilon}{\sigma_E} \right],
    \end{equation}
    and
    \begin{equation}
        \label{eq:20}
        \braket{D_{P_\hk}(\rho_t,\omega) }_{T \rightarrow \infty} \leq 2 \sqrt{\frac{K}{{\de}}} \ll 1
    \end{equation}
    (it is above some constant for long times, but still equilibrates
    eventually).
\end{theorem}
\begin{proof}
    \label{pr:2}
    Since $\tau$ is a very small time step, the overlap between
    $\ket{\psi(0)}$ and $\ket{\psi(\tau)}$ is nearly $1$. To prove
    this we write $\ket{\psi(t)}$ in the energy basis
    \begin{equation}
        \ket{\psi(t)} = \sum_{n}^{\tilde{d}} c_n e^{-i E_n t} \ket{n}
    \end{equation}
    and calculate its internal product with its initial state
    \begin{align}
        &\left| \braket{\psi(t) \middle| \psi(0)} \right|^2
        = \left| \sum_n^{\tilde{d}} |c_n|^2 e^{-i E_n t} \right|^2 \notag\\
        &\quad= \sum_{n m}^{\tilde{d}} |c_n|^2 |c_m|^2 \cos\left( (E_n-E_m)t \right) \notag\\
        &\quad\ge 1 - \frac{t^2}{2} \sum_{n m}^{\tilde{d}} |c_n|^2 |c_m|^2 \left( E_n^2\! + \!E_m^2\! -\! 2E_nE_m \right) \notag\\
        &\quad= 1 - \left( 2\overline{E^2} - 2\overline{E}^2 \right)\frac{t^2}{2} \notag\\
        &\quad= 1 - \sigma_E^2 t^2,
    \end{align}
    where $\sigma_E$ is the standard deviation in energy. So, we have that
    \begin{align}
        \left| \braket{\psi(t) \middle| \psi(0)} \right|^2
        \geq 1 - \epsilon^2,
    \end{align}
    $\forall t$ such that $\left| t \right| \leq \tau/2 =
    \epsilon/\sigma_E$. This trivially implies
    \begin{equation}
        \left| \braket{\psi(t) \middle| \psi(t')} \right|^2
        \geq 1 - \epsilon^2,
    \end{equation}
    $\forall t,t'$ such that $\left| t - t' \right| \leq \tau/{2}$.

    Meanwhile, $\hk$ contains, by definition, all projectors
    $\ket{\psi(j\tau)}\!\bra{\psi(j\tau)}$ for $j$ up to $K-1$.
    Therefore, for any time $t$ up to $(K-\,^1\!/_2)\tau$, the state
    $\ket{\psi(t)}$ is very close to one of these projectors.

    In other words, there is always a value of $0 \leq j \leq K-1$ such
    that $|t - j\tau| \leq \tau/2$ and
    \begin{align}
        \tr{\rho_t P_{\hk}}
        &= \bra{\psi(t)} \left[ P_{j\tau} + P_{j\tau}^\perp \right] \ket{\psi(t)} \notag\\
        &\geq \left| \braket{\psi(j\tau) \middle| \psi(t)} \right|^2 \notag\\
        &\ge 1 - \epsilon^2,
    \end{align}
    where $P_{t} = \ket{\psi(t)}\bra{\psi(t)}$ and $P_{t}^\perp = P_\hk -
    P_{t}$. This directly leads to eq.~(\ref{eq:9}),
    \begin{align}
        \label{eq:21}
        D_{P_\hk}(\rho_t,\omega) &= \left| \tr{P_{\hk}(\rho_t-\omega)} \right| \notag\\
        &\geq \tr{P_{\hk}\rho_t} - \tr{P_{\hk}\omega} \notag\\
        &\geq 1 - \epsilon^2 - \sqrt{\frac{K}{\de} }.
    \end{align}
    Equation~(\ref{eq:20}) is easily obtained from the Cauchy-Schwarz
    inequality
    \begin{align}
        \label{eq:25}
        &\braket{D_{P_\hk}(\rho_t,\omega) }_{T \rightarrow \infty} \notag\\
        &\qquad\quad= \big\langle{\left| \tr{{P_\hk}(\rho_t-\omega)} \right|}\big\rangle_{T \rightarrow \infty} \notag\\
        &\qquad\quad\leq \big\langle{ \tr{{P_\hk}\rho_t}
          + \tr{{P_\hk}\omega}}\big\rangle_{T \rightarrow \infty} \notag\\
        &\qquad\quad= 2\tr{{P_\hk}\omega} \leq 2 \sqrt{\tr{{P_\hk}^2}\tr{\omega^2}} \notag\\
        &\qquad\quad\leq 2 \sqrt{\frac{K}{{\de}}} \ll 1.
    \end{align}
\end{proof}

The role played by equations~(\ref{eq:9}) and~(\ref{eq:20}) is
simple. \emph{(i)} The system obviously has not equilibrated, and is
still distinguishable from its equilibrium state, as long as
$D_{P_\hk}(\rho_t,\omega) $ is significantly above zero. \emph{(ii)}
On the other hand, the rank of $P_\hk$ is small enough that any system
spread over many energy levels will equilibrate with
respect to it.

For instance, if we take $K = \de/1000$ (which is extremely large) and
$\epsilon=1/2$, we have
\begin{equation}
    \label{eq:14}
    D_{P_\hk}(\rho_t,\omega) \geq \frac{1}{2} ,
    \quad \forall t \in \left[ 0,\frac{\de}{1000\sigma_E} \right],
\end{equation}
For systems composed of many particles we would typically expect $\sigma_E \sim \log(\de)$ leading to the
system taking a time of order $\frac{\de}{\log(\de)}$ to
equilibrate with respect to this measurement.

To illustrate how large this time scale can be, we describe now a
simple example -- the time scale depends only on $\sigma_E$ and $\de$,
and is largely independent on the details. Consider a system of $L$
weakly-interacting qubits with level-spacing $\delta E = 10^{-18}
\mathrm{Joules}$, the order of the excitation energy in atoms.
Defining each qubit to have equal population on each level, simple
calculations give $\sigma_E \approx \sqrt{L} \delta E $ and $\de
\approx 2^L$, and we get $T^\mathrm{slow}_\text{eq} > \frac{\hbar
  \de}{1000\sigma_E} \approx {2^L L^{-\frac{1}{2}} 10^{-19}
  \second}$~\footnote{Throughout the paper, we choose units such that
  $\hbar = 1$. Only in this example we adopt S.I. units for
  illustration purposes.}. Then, taking as little as $125$ qubits
already gives $T^\mathrm{slow}_\text{eq} \gtrsim 4.10^{17} \second$,
nearly the age of the universe and increasing exponentially with $L$.
In contrast, for the same number of particles, the
average distinguishability of a typical measurement falls below $10^{-3}$ in a
time scale of $T_\text{eq}^{\mathrm{typ}} \lesssim \frac{6000^2 \hbar
}{\sigma_E} \approx 3 \times 10^{-10} \second$. This typical time
scale decreases with ${L^{-\frac{1}{2} }}$, becoming even smaller for macroscopic
systems, and is obtained from
Theorem~\ref{thr:5} by assuming $\eta_\frac{1}{T} \lesssim 1/ \sigma_E
T$ as discussed in the main text.

Of course, the construction in Theorem~\ref{thr:1} is not the only
possibility and indeed an alternative construction is given in \cite{Goldstein13}. For instance, one can easily define measurements with a
larger number of outcomes, which also obey eq.~(\ref{eq:9}) for at
least as long as $D_{P_\hk}$ (see appendix~\ref{sec:n-outc-meas}).
It is also worth mentioning that this theorem trivially extends to the
existence of an observable and whose expectation value takes a long
time to equilibrate, since $P_\hk$ is, of course, an observable. The
distinguishability simply presents a stronger definition of
equilibration.

\section{Extension to $N$-Outcomes}
\label{sec:n-outc-meas}

Theorem~\ref{thr:4} can be generalized to $N$-outcome measurements
$\cm = \set{P_1,\ldots,P_N}$ with the bound
$\braket{D_\cm(\rho_t,\omega) }_T \leq
\frac{c}{2}\sqrt{\eta_{\frac{1}{T} }} \sum_{i =1 }^{N} \sqrt{k_i}$
where $k_i = \min{\{\rank{P_i},d-\rank{P_i}\}}$.

There are several ways one could extend Theorem~\ref{thr:1}. One is to
simply divide $\hk$ into $N-1$ smaller subspaces. Then one has
$\cm=\set{P_{\hk_1},\ldots,P_{\hk_{N-1}},\id-P_\hk}$, and the
resulting distinguishability
\begin{align}
    \label{eq:36}
    D_{\cm}(\rho_t,\omega) &= \frac{1}{2} \sum_{n = 1}^{N-1} \left| \tr{P_{\hk_n}(\rho_t-\omega)} \right| \notag\\
    &\quad\quad+ \frac{1}{2} \left| \tr{(\id-P_{\hk})(\rho_t-\omega)} \right|\notag\\
    &\geq \frac{1}{2} \left| \tr{\sum_{n = 1}^{N-1} P_{\hk_n}(\rho_t-\omega)} \right| \notag\\
    &\quad\quad+ \frac{1}{2} \left| \tr{(\id-P_{\hk})(\rho_t-\omega)} \right|\notag\\
    &= D_{P_\hk}(\rho_t,\omega)
\end{align}
takes at least as long to equilibrate as $D_{P_\hk}$.

% \bibliography{references}

\end{document}